\pdfoutput=1

\documentclass[final,twoside,11pt]{entics} 
\usepackage{enticsmacro}
\usepackage{preamble}
\def\conf{MFPS 2024}
\def\myconf{mfps40} 	
\volume{4}			
\def\volu{4}			
\def\papno{14}			
\def\mydoi{14767}%
\def\lastname{Kavvos} 
\def\runtitle{Two-dimensional Kripke Semantics II} 
\def\copynames{G.A. Kavvos}    
\begin{document}
\begin{frontmatter}
  \title{Two-dimensional Kripke Semantics II: \\[1ex] Stability and Completeness\thanksref{ALL}}
 \thanks[ALL]{
  This work was supported by the UKRI Engineering and Physical Sciences
  Research Council grants EP/Y000242/1 and EP/Y033418/1, as well as a grant from the
  Advanced Research + Invention Agency (ARIA).
  }
  \author{G. A. Kavvos\thanksref{a}\thanksref{alexemail}}	
   \address[a]{School of Computer Science\\ University of Bristol\\ Bristol, United Kingdom}
   \thanks[alexemail]{Email: \href{mailto:alex.kavvos@bristol.ac.uk} {\texttt{\normalshape
        alex.kavvos@bristol.ac.uk}}} 

\begin{abstract} 
  We revisit the duality between Kripke and algebraic semantics of
  intuitionistic and intuitionistic modal logic. We find that there is a certain
  mismatch between the two semantics, which means that not all algebraic models
  can be embedded into a Kripke model. This leads to an alternative proposal for
  a relational semantics, the stable semantics. Instead of an arbitrary partial
  order, the stable semantics requires a distributive lattice of worlds. We
  constructively show that the stable semantics is exactly as complete as the
  algebraic semantics. Categorifying these results leads to a 2-duality between
  two-dimensional stable semantics and categories of product-preserving
  presheaves, i.e. models of algebraic theories in the style of Lawvere.
\end{abstract}

\begin{keyword}
  intuitionistic logic, modal logic, intuitionistic modal logic, Kripke
  semantics, algebraic semantics, duality, filters, presheaves, sifted colimits,
  product-preserving presheaves, Lawvere theories
\end{keyword}

\end{frontmatter}

\section{Introduction}
  \label{section:introduction}

In a previous paper I revisited the relationship between the Kripke and
algebraic semantics of intuitionistic logic and (intuitionistic) modal logic
\cite{kavvos_2024a}. Kripke frames (i.e. partial orders) correspond to a certain
class of complete Heyting algebras, the \emph{prime algebraic lattices}. This
amounts to a \emph{duality} $\Op{\POSET} \Equiv \PRIMEALGLATT$, which may be
refined to `truth-preserving' morphisms on one side, and implication-preserving
on the other.

What is curious is that this duality can be reproduced at the level of
categories, which model \emph{proofs}. Replacing a Kripke frame by a category
leads to an evident definition of a proof-relevant \emph{two-dimensional Kripke
semantics}. Taking presheaves over the category gives a \emph{bicartesian closed
category}, i.e. a categorical/proof-relevant model of intuitionistic logic. The
interpretation of formulas is then a direct categorification of Kripke
semantics. This leads to a 2-duality $\textbf{Cat}^\text{op}_\text{cc} \Equiv
\PSHCAT$ between Cauchy-complete categories (qua two-dimensional Kripke frames)
and presheaf categories (qua prime algebraic lattices). 

Moreover, this story can be adapted to \emph{intuitionistic modal logic}. There
is no widespread agreement on what the latter is. However, in
\cite{kavvos_2024a} I showed that a relation that is canonically compatible with
a partial order (i.e. a \emph{bimodule}) induces two adjoint modalities
$\blacklozenge \Adjoint \Box$ by Kan extension. This provides a canonical
proposal as to what an intuitionistic modal logic should be; its corresponding
Kripke semantics is
\begin{align}
  \label{equation:kripke-semantics-modal}
  \KSat{w}{\blacklozenge \varphi}\
  &\defequiv\
    \exists v.\ v \mathbin{R} w \text{ and } \KSat{v}{\varphi}
  \\
  \KSat{w}{\Box \varphi}\
  &\defequiv\
    \forall v.\ w \mathbin{R} v \text{ implies } \KSat{v}{\varphi}
\end{align}
While $\Box$ is indeed the expected modality, $\blacklozenge$ uses $R$ in the
\emph{opposite} variance to the more common $\lozenge$ modality. Conversely, any
adjunction $\blacklozenge \Adjoint \Box$ on a prime algebraic lattice uniquely
corresponds to a bimodule, giving a duality $\Op{\EBIMOD} \Equiv \PRIMEALGLATTO$
between bimodules on a partial order and prime algebraic lattices $L$ equipped
with an \emph{operation}, i.e. a meet-preserving $\Box : L \to L$ (equivalently,
a join-preserving $\blacklozenge : L \to L$).

The modal picture can also be categorified, by adding proofs to bimodules,
making them into \emph{profunctors}. Left Kan extension then induces an
adjunction on $\FUNC{\CC}{\SET}$. By unfolding the definition of these adjoints
we obtain a remarkable proof-relevant version of
\eqref{equation:kripke-semantics-modal}, viz. the presheaves
\begin{align*}
  (\blacklozenge P)(w) &= \int^{v \in \CC} R(v, w) \times P(v) \\
  (\Box P)(w) &= \int_{v \in \CC} R(w, v) \to P(v)
  \cong \Hom[\FUNC{\CC}{\SET}]{R(w, -)}{\sem{\varphi}}
\end{align*}
for each presheaf $P : \CC \fto \SET$ (qua formula). This again amounts to a
2-duality $\EPROF^\text{op}_\text{cc} \Equiv \PSHCATO$ between profunctors on a
Cauchy-complete category, and presheaf categories equipped with a continuous
$\Box : \FUNC{\CC}{\SET} \fto \FUNC{\CC}{\SET}$ (equivalently, a cocontinuous
$\blacklozenge : \FUNC{\CC}{\SET} \fto \FUNC{\CC}{\SET}$, by local
presentability). This is consistent with what we have come to regard in the last
few years as the categorical semantics of modal logic, i.e. an adjunction on a
bicartesian closed category
\cite{clouston_2018,birkedal_2020,gratzer_2020,gratzer_2021,gratzer_2022}.

\subsection*{Completeness} 

These dualities also come with theorems relating validity in the Kripke
semantics to validity in the induced algebraic semantics. Consequently, we are
able to use them to prove completeness of the algebraic semantics from
completeness of the Kripke semantics, as follows. Suppose that a formula of
intuitionistic logic is valid in all Heyting algebras; it is then a fortiori
valid in all prime algebraic lattices, and hence valid in all Kripke frames.
Therefore, if the Kripke semantics is complete, this formula must be provable.
As a result, completeness of the Kripke semantics implies completeness of the
algebraic semantics.

Surprisingly, the converse implication is also provable. An old construction,
whose origins we can trace at least as far as the book by Fitting \cite[\S
1.6]{fitting_1969}, gives a recipe for inducing a Kripke semantics from a
general Heyting algebra, by taking all \emph{prime filters}. The resulting
structure is richer than an ordinary frame: it is a \emph{descriptive frame}
\cite[\S 8.4]{chagrov_1996}. This is part of a duality between Heyting algebras
and descriptive frames, which is known as \emph{Esakia duality}
\cite{esakia_2019}. It is then possible to relate validity in the descriptive
frame to validity in the Heyting algebra. A categorical version of this
construction for coherent toposes has been shown by Joyal: see \cite[Theorem
6.3.5]{makkai_1977}. When simplified, Joyal's result amounts to an embedding of
every Heyting algebra into a prime algebraic lattice that preserves all
connectives \cite[\S 3.2]{harnik_1992} \cite{makkai_1995} \cite{galli_2000}.

However, the part of this result that relates validity in the descriptive frame
to validity in the Heyting algebra requires the \emph{prime filter existence
theorem} \cite[\S 10]{davey_2002} \cite[\S I.2.3]{johnstone_1982}, which is a
weak form of the axiom of choice. Similarly, the result of Joyal quoted above
uses highly non-constructive reasoning.

\emph{This paper is about trying to avoid this particular reasoning step.}
Unlike other streams of work \cite{bezhanishvili_2020}, this is not due to a
predilection for constructivity. Instead, we are looking for a construction that
we can \emph{categorify}, so that it applies to models of intuitionistic (modal)
proofs as well. This will in turn provide interesting information about the
completeness of various classes of models of typed (modal) $\lambda$-calculi.

\subsection*{Stable semantics}

However, relating Kripke and algebraic semantics appears impossible without
using prime filters. In an attempt to overcome this I will introduce a new
relational semantics for intuitionistic logic, which I call \emph{stable
semantics}. The essence can summarised as replacing upper sets, which play a
central r\^{o}le in Kripke semantics \cite{kavvos_2024a}, with \emph{filters}.
This inescapably leads to the use of distributive lattices as `Kripke frames,'
as well as a different interpretation of disjunction, which is reminiscent of
Beth semantics \cite{bezhanishvili_2019} and Kripke-Joyal semantics \cite[\S
II.9]{lambek_1988} \cite[\S VI.6]{mac_lane_1994} \cite[\S 6.6]{borceux_1994c}.
The attendant duality, which is now between distributive lattices and
\emph{spectral locales} (aka coherent frames), is already well-known from Stone
duality \cite[\S II.3.2]{johnstone_1982}. Furthermore, the coherent semantics
can be straightforwardly extended to modalities.

The advantage of stable semantics is that we can constructively show an
\emph{equi-completeness} result. Every Heyting algebra is a distributive
lattice, and so can be used as the `possible worlds' of a stable semantics.
Moreover, every distributive lattice can be interpreted into a complete Heyting
algebra---in fact a spectral locale---in a way that preserves all logical
structure. Thus, the completeness of the stable semantics directly follows from
the completeness of the algebraic semantics.

\subsection*{Two-dimensional stable semantics and algebraic theories}

Categorifying the above story engenders a surprise. For our purposes, the most
technically expedient categorification of filters is the \emph{sifted colimit
completion}. The very same completion plays a central r\^{o}le in the
\emph{algebraic theories} in the style of Lawvere \cite{adamek_2010}: the
category of algebras is just the sifted completion of the opposite of its
\emph{theory}, which is just a cartesian category (i.e.~has finite products).

If we assume that the opposite of a theory is a \emph{distributive category}
\cite{cockett_1993}, the results on stable semantics can be directly
categorified. This shows that the class of product-preserving presheaves (cf.
filters) on that distributive category (cf. stable frame) is a complete model of
the typed $\lambda$-calculus with sums and an empty type. These results can be
readily adapted to the proofs of intuitionistic modal logic.

As a result, the results herein bear a striking kinship with those of
categorical algebra. I am not yet certain what the long-term impact of this
observation is, but it seems far too compelling to ignore.

\subsection*{Roadmap}

In \S\ref{section:int-models-as-frames} I discuss what it means to regard a
Heyting algebra as a set of possible worlds, as well as the technical issues
that arise when we try to embed that representation into a prime algebraic
lattice. This leads to the introduction of stable semantics in
\S\ref{section:stable}, which is proved equi-complete with Heyting algebras. In
\S\ref{section:stable-modal} I show that the stable semantics can be
effortlessly adapted to interpret adjoint modalities. Then, in
\S\ref{section:2dss-int} I proceed to categorify the stable semantics; this
requires a recap of the elements of Lawvere's approach to algebraic theories.
Following that I give an equi-completeness proof, and discuss the resultant
syntax-semantics 2-duality. Finally, this approach is extended to intuitionistic
modal proofs in \S\ref{section:2dss-modal}.

\section{Heyting algebras as possible worlds}
  \label{section:int-models-as-frames}

It is a folklore fact that every Kripke frame $(W, \sqsubseteq)$ induces a prime
algebraic lattice $\FUNC{W}{\TV}$, which consists of its upper sets ordered by
inclusion; see the first paper on two-dimensional Kripke semantics for a
detailed exposition \cite[\S 2]{kavvos_2024a}. Looking at this lattice as a
Heyting algebra, i.e. an algebraic semantics for intuitionistic logic, we see
that every formula $\varphi$ is interpreted as the set $\sem{\varphi} \subseteq
W$ of worlds in which it is true. This set is upper because Kripke semantics is
monotonic: $w \sqsubseteq v$ can be read as saying that world $v$ has
potentially more information than world $w$. Thus, the passage from $w$ to $v$
may force more formulas to be true, but will not invalidate formulas that were
previously known to be true.

It is interesting to consider a Heyting algebra $H$ in the capacity of a Kripke
frame itself. The most evident way of doing so is by taking the opposite of its
order, yielding the partial order $(\Op{H}, \sqsubseteq)$, where $\sqsubseteq$
is just $\geq$ in $H$. Thinking of $H$ as a Tarski-Lindenbaum algebra of an
intuitionistic theory, we see that
\[
  \varphi \sqsubseteq \psi 
    \quad\text{ iff }\quad
  \psi \leq \varphi
    \quad\text{ iff }\quad
    \text{``}\psi \vdash \varphi\text{''}
\]
Roughly, each element $\varphi \in H$ can be thought of as a formula that
specifies what we currently know. The relation $\varphi \sqsubseteq \psi$ holds
just when $\psi$ implies $\varphi$, i.e. when $\psi$ potentially contains more
information.

The order-embedding
$
  \UpSet : H \to \FUNC{\Op{H}}{\TV}
$
then takes $\varphi \in H$ to $\SetComp { \psi \in H }{ \varphi \leq \psi } =
\SetComp{ \psi \in H }{ \psi \leq \varphi }$, i.e. the set of formulas that
imply $\varphi$. It is well-known that $\UpSet$ preserves finite meets and
exponentials, so that
\begin{align*}
  \UpSet{\top} &= H
                &
  \UpSet*{\varphi \land \psi} &= \UpSet{\varphi} \land \UpSet{\psi}
                             &
  \UpSet*{\HeyExp{\varphi}{\psi}} &= \HeyExp{\UpSet{\varphi}}{\UpSet{\psi}}
\end{align*}
However, $\UpSet$ famously does \emph{not} preserve disjunction: sometimes
$\UpSet*{\varphi \lor \psi} \neq \UpSet{\varphi} \lor \UpSet{\psi}$. Thus, we
can only embed the $(\land \to)$ fragment of the logic into a prime algebraic
lattice in this manner.

These facts are perhaps better known at the two-dimensional level. Suppose that
$\CC$ is a bicartesian closed category, i.e. a model of intuitionistic proofs.
It is a basic fact of category theory that the Yoneda functor $\Yo : \CC \fto
\FUNC{\Op{\CC}}{\SET}$ is an \emph{embedding}, i.e. full, faithful, and
injective on objects. It is also well-known that $\Yo$ preserves finite products
and exponentials \cite{awodey_2010}, i.e. that
\begin{align*}
  \Yo{\terminal} &\cong \terminal
                 &
  \Yo{c \times d}  &\cong \Yo{c} \times \Yo{d}
                   &
  \Yo{\HeyExp{c}{d}} &\cong \HeyExp{\Yo{c}}{\Yo{d}}
\end{align*}
For a totally unrelated purpose, Dana Scott \cite{scott_1980} noticed that this
induces a useful isomorphism:
\begin{lemma}[Scott]
  \label{lemma:ccc-conservativity-1}
  If $\varphi$ uses neither disjunction nor falsity then
  $\sem{\varphi}_{\FUNC{\Op{\CC}}{\SET}} \cong \Yo{\sem{\varphi}_{\CC}}$.
\end{lemma}
Here $\sem{\varphi}_{\CC}$ is the interpretation of $\varphi$ as an object of
$\CC$, and $\sem{\varphi}_{\FUNC{\Op{\CC}}{\SET}}$ is the interpretation of
$\varphi$ as an object of the category of presheaves $\FUNC{\Op{\CC}}{\SET}$,
both defined following the respective cartesian closed structure. In the second
instance basic propositions $p$ are interpreted by the representable
$\Yo{\sem{p}_{\CC}}$.

It is not difficult to extend this result to categorical models of modal logic.
Following the work of Clouston on Fitch-style $\lambda$-calculi
\cite{clouston_2018}, these are generally understood to be endo-adjunctions
\begin{diagram}
  \label{diagram:fitch-style}
  \begin{tikzpicture}[node distance=2.5cm, on grid]
    \node (W) {$\CC$};
    \node (V) [right = 4cm of W] {$\CC$};
    \path[->, bend left=20] (W) edge node [above] {$\Box$} (V);
    \path[->, bend left=20] (V) edge node [below] {$\blacklozenge$} (W);
    \node (Adj1) [right = 2cm of W] {$\AdjointLDown$};
  \end{tikzpicture}
\end{diagram}
on a bicartesian closed category $\CC$. Given such a model, take the left Kan
extension of $\Yo \circ \blacklozenge$ along Yoneda:
\begin{equation}
  \label{diagram:left-kan-scott}
  \begin{tikzpicture}[node distance=2.5cm, on grid, baseline=(current  bounding  box.center)]
    \node (C) {$\CC$};
    \node (PSHC) [right = 6cm of C] {$\FUNC{\Op{\CC}}{\SET}$};
    \node (E) [below = 2cm of PSHC] {$\FUNC{\Op{\CC}}{\SET}$};
    \node (C2) [below = 2cm of C] {$\CC$};
    \path[->] (C) edge node [above] {$\Yo$} (PSHC);
    \path[->] (C2) edge node [below] {$\Yo$} (E);
    \path[->] (C) 
      edge 
        node [left] {$\blacklozenge$} 
        node [right = .05em] {$\Adjoint$}
      (C2);
    \path[->, bend right = 45] (C2) edge node [right] {$\Box$} (C);
    \path[->, dashed] (PSHC) 
      edge 
      node [left] {$\blacklozenge_{p}$} 
        node [right = .15em] {$\Adjoint$}
      (E);
      \path[->, bend right = 45, dotted] (E) edge node [right] {$\Box_p$} (PSHC);
  \end{tikzpicture}
\end{equation}
$\blacklozenge_p$ is then a colimit-preserving functor on the presheaf category,
and has a right adjoint $\Box_p$. Thus, we obtain a categorical model of modal
logic on the presheaf category. It is easy to calculate that the action of these
adjoint functors on representables is essentially the same as that of
$\blacklozenge$ and $\Box$, in that \eqref{diagram:left-kan-scott} commutes:
\begin{align*}
  \blacklozenge_p\prn{\Yo{c}} &\defeq \Yo{\blacklozenge c} \\
  \Box_p\prn{\Yo{c}} 
    &\defeq \Hom[\FUNC{\Op{\CC}}{\SET}]{\DelimMin{1} \Yo{\blacklozenge(-)}}{\Yo{c}}
     \cong  \Hom[\CC]{\blacklozenge(-)}{c}
     \cong  \Hom[\CC]{-}{\Box c}
         =  \Yo{\Box c}
\end{align*}
Consequently, Scott's lemma directly extends to the categorical semantics of the
$(\land \to \blacklozenge \Box)$ fragment of intuitionistic modal logic. Notice
that the diagram \eqref{diagram:left-kan-scott} witnesses $\Yo$ as a (weak)
morphism of categorical models of modal logic without disjunction: $\Yo$ is a
cartesian closed functor that preserves the adjunction. Of course, this result
can be de-categorified to one for Heyting algebras equipped with an adjunction.

This leaves the mystery of disjunction. One might think that sheaves are the
answer. However, we will do something more radical instead.

\section{Stable semantics of intuitionistic logic}
  \label{section:stable}

Given an arbitrary Kripke frame, i.e. a partial order $(W, \sqsubseteq)$, Kripke
semantics interprets every formula as an \emph{upper set} of worlds, i.e. a set
$S \subseteq W$ for which $w \in S$ and $w \sqsubseteq v$ implies $v \in S$. The
stable semantics will instead revolve around the notion of a \emph{filter} over
$W$.
\begin{definition}
  A \emph{filter} over $(W, \sqsubseteq)$ is a non-empty subset $F \subseteq W$
  which is
  \begin{itemize}
    \item \emph{upper}, in that $w \in F$ and $w \sqsubseteq v$ implies $v \in
      F$; and
    \item \emph{filtered}, in that whenever $w, v \in F$ there exists a $z \in
      F$ with $z \sqsubseteq w$ and $z \sqsubseteq v$.
  \end{itemize}
\end{definition}
We write $\Filt*{W}$ for the set of filters over $W$. $\Filt*{W}$ is a poset
under inclusion---in fact it is a \emph{directed complete partial order}
(dcpo) (without a bottom element) \cite[\S O-2.8]{gierz_2003}.

When $W$ has more structure the definition of a filter can be somewhat
simplified.
\begin{proposition}
  \label{proposition:filters-meet-semilattice}
  Let $(W, \sqsubseteq)$ be a meet-semilattice. A subset $F \subseteq W$ is a
  filter if and only if it is
  \begin{itemize}
    \item \emph{upper}, in that $w \in F$ and $w \sqsubseteq v$ implies $v \in
      F$; and
    \item a \emph{sub-meet-semilattice}, in that $1 \in F$ and $w, v \in F$
      implies $w \land v \in F$
  \end{itemize}
\end{proposition}

A \emph{stable frame} is a partial order $(W, \sqsubseteq)$ which is a
\emph{distributive lattice}. This means that it has both finite joins and meets,
and that they also satisfy the distributive law $a \land (x \lor y) = (a \land
x) \lor (a \land y)$. Consequently, they also satisfy the dual law $a \lor (x
\land y) = (a \lor x) \land (a \lor y)$ \cite[\S I.1.5]{johnstone_1982}, which
we will use heavily.

A stable frame has much more structure than a good old fashioned Kripke frame.
To begin, any two worlds $w, v \in W$ have a meet $w \land v$ and a join $w \lor
v$ (we use the same notation as the logic, but rely on context for
disambiguation). If we think of each world as containing information---in
particular about which variables have become true---then these two operators
tell us that it is possible to find least and greatest upper bounds of
information. The fact the distributive law holds means that the interpretation
of these bounds as `intersection of information' and `union of information' is
tenable.

Furthermore, $W$ has a bottom element $0$ and a top element $1$. The bottom
element $0$ represents the \emph{baseline level of information}, i.e. the fewest
facts we may regard as true. The top element $1$ represents a \emph{supernova of
information}. This amount of information implies all formulae---even false ones.

A \emph{stable model} $\Model{M} = (W, \sqsubseteq, V)$ consists of a stable
frame $(W, \sqsubseteq)$ and a function $V : \Vars \to \Filt*{W}$. The
\emph{valuation} $V$ assigns to each propositional variable $p \in \Vars$ a
filter $V(p) \subseteq W$, to be thought of as the set of worlds in which $p$ is
true. The fact this is a filter leads to the following intuitions:
\begin{description}
  \item[Upper set] A proposition that is true remains true as information
    increases.
  \item[Top element] Every proposition is true at the supernova world $1$.
  \item[Meets] If $w, v \in V(p)$ then both $w$ and $v$
    contain the information that $p$ is true. Therefore, their greatest lower
    bound should also contain that information, so that $w \land v \in V(p)$.
\end{description}
Notice that if $0 \in V(p)$ then $V(p) = W$, as filters are upper sets. Thus, a
variable that is true at the baseline world $0$ is true throughout a stable
frame.

The \emph{stable semantics} is defined through a relation
$\KSat[\Model{M}]{w}{\varphi}$ with the meaning that $\varphi$ is true in world
$w$ of model $\Model{M}$. When it is clear which model we are using we will skip
it, writing simply $\KSat{w}{\varphi}$. The clauses for
$\KSat[\Model{M}]{w}{\varphi}$ are much like those for the Kripke semantics,
with the characteristic clause for $\to$:
\[
  \KSat[\Model{M}]{w}{\varphi \to \psi}\
  \defequiv\
    \forall w \sqsubseteq v.\
    \KSat[\Model{M}]{v}{\varphi} \text{ implies } \KSat[\Model{M}]{v}{\varphi}
\]
The only clauses that change are the ones for falsity and
disjunction:\footnote{One of the reviewers pointed out that the clause for
disjunction can be simplified to $w = v_1 \land v_2$. However, that version of
the clause does not seem immediately amenable to categorification, unlike this
one.}
\begin{alignat*}{3}
  &\KSat[\Model{M}]{w}{\bot}\
  &&\defequiv\
    (w = 1)
  \\
  &\KSat[\Model{M}]{w}{\varphi \lor \psi}\
  &&\defequiv\
    \exists v_1, v_2.\,\, v_1 \land v_2 \sqsubseteq w \text{ and }
    \KSat[\Model{M}]{v_1}{\varphi} \text{ and }
    \KSat[\Model{M}]{v_2}{\psi}
\end{alignat*}
There are a number of things to notice about this definition.

First, the falsity $\bot$ can now be a true formula; but it is only true at $1
\in W$, which is top element for the information order $\sqsubseteq$. In fact,
every formula is true at $1$. In that sense, $1$ is a world that contains so
much information that it forces everything---even falsity!---to be true. A
similar concept of \emph{exploding}, \emph{fallible} or \emph{inconsistent
worlds} is common in the context of intuitionistically-provable completeness
proofs for intuitionistic logic and associated realizability models
\cite{veldman_1976,de_swart_1976,troelstra_1988,harnik_1992,lipton_1992}.

Second, the clause for the disjunction $\varphi \lor \psi$ at world $w$ requires
that \emph{both} $\varphi$ and $\psi$ are true at some worlds $v_1$ and $v_2$
respectively. However, the common information between $v_1$ and $v_2$, i.e. $v_1
\land v_2$, must be less than the information at $w$. When $v_1 \land v_2
\sqsubseteq w$ we say that $w$ \emph{fans into} $v_1$ and $v_2$. Alyssa Renata
has proposed an intuitionistic reading of this clause: the world $w$ forces
$\phi \lor \psi$ exactly when it consistent to develop into a state of mind
where the mathematician has proven $\phi$, as well as a state of mind where they
have proven $\psi$. But what if one of the two formulas in a disjunction is a
contradiction? In that case we are saved by the supernova world: as $w \land 1 =
w$, we have that $\KSat{w}{\varphi \lor \bot}$ if and only if
$\KSat{w}{\varphi}$.

Third, note that the definition does not mention the joins $w \lor v$ of worlds
of $W$, even though such joins exist. While joins do not appear explicitly in
the semantic clauses, they are used in the following lemma to show that the
stable semantics is monotonic---in particular in the case of
implication.\footnote{I am grateful to an anonymous reviewer for pointing out
  that, strictly speaking, we do not need joins: we could just as well reproduce
  this result even in the weaker setting of a \emph{distributive semilattice},
  i.e. a meet-semilattice with the property that $a \land b \sqsubseteq x$
  implies that $x = a' \land b'$ for some $a \sqsubseteq a'$ and $b \sqsubseteq
  b'$. However, we prefer having joins, as their categorification is
  well-understood.}

\begin{lemma}[Filtering] \hfill
  \label{lemma:filtering}
  \begin{enumerate}
    \item $\KSat[\Model{M}]{w}{\varphi}$ and $w \sqsubseteq v$ imply
      $\KSat[\Model{M}]{v}{\varphi}$.
    \item $\KSat[\Model{M}]{1}{\varphi}$ for any $\varphi$.
    \item $\KSat[\Model{M}]{w_1}{\varphi}$ and $\KSat[\Model{M}]{w_2}{\varphi}$
      imply $\KSat[\Model{M}]{w_1 \land w_2}{\varphi}$.
  \end{enumerate}
\end{lemma}

\begin{proof}
  We prove (iii), and only show the cases for implication and disjunction.

  Suppose that $\KSat{w_1}{\varphi \to \psi}$, $\KSat{w_2}{\varphi \to \psi}$,
  $w_1 \land w_2 \sqsubseteq v$ and $\KSat{v}{\varphi}$. As $w_i \sqsubseteq w_i
  \lor v$ and $v \sqsubseteq w_i \lor v$ we know that $\KSat{w_i \lor v}{\varphi
  \to \psi}$ and $\KSat{w_i \lor v}{\varphi}$ by (i), and hence that $\KSat{w_i
  \lor v}{\psi}$, for $i \in \{ 1, 2 \}$. Hence, by the IH, $\KSat{(w_1 \lor v)
  \land (w_2 \lor v)}{\psi}$. But $v = (w_1 \land w_2) \lor v = (w_1 \lor v)
  \land (w_2 \lor v)$ by distributivity.

  Suppose that $\KSat{w_1}{\varphi_1 \lor \varphi_2}$ and $\KSat{w_2}{\varphi_1
  \lor \varphi_2}$. Then there exist $v_{ij}$ with $v_{i1} \land v_{i2}
  \sqsubseteq w_i$ and $\KSat{v_{ij}}{\varphi_j}$. Then $\KSat{v_{1j} \land
  v_{2j}}{\phi_j}$, with $(v_{11} \land v_{21}) \land (v_{21} \land v_{22}) =
  (v_{11} \land v_{12}) \land (v_{12} \land v_{22}) \sqsubseteq w_1 \land w_2$.
\end{proof}

Both (i) and (iii) of this lemma require the existence of disjunctions; in fact,
they make essential use of the \emph{dual} distributive law $a \lor (x \land y)
= (a \lor x) \land (a \lor y)$.

It now remains to show how the stable semantics induce an algebraic semantics.
Given a stable frame $(W, \sqsubseteq)$ consider the set $\FUNC{W}{\TV}[\land]$
of monotonic functions $p : W \to \TV$ which preserve finite meets. This is a
partial order under the pointwise order. This poset has a number of curious
properties.

First, the monotonicity of $p : W \to \TV$ implies that if $p(w) = 1$ and $w
\sqsubseteq v$, then $p(v) = 1$. Hence, the subset $U = \ReIdx{p}{1}$ of $W$ is
an upper set. As $p(\top) = 1$, we know that $\top \in U$. Moreover, if $p(w) =
1$ and $p(v) = 1$, then $p(v \land w) = p(v) \land p(w) = 1$, so $U$ is closed
under finite meets. In short, $U$ is a filter. It is not difficult to show that
every filter $F \subseteq W$ gives rise to a map $p_F : W \to \TV$ which is
monotonic and finite-meet-preserving. Consequently, there is an order-bijection
$\Filt*{W} \cong \FUNC{W}{\TV}[\land]$. I will keep using the somewhat
cumbersome notation $\FUNC{W}{\TV}[\land]$ for reasons that will become clear
later.


Second, the poset $\FUNC{W}{\TV}[\land]$ is a \emph{complete lattice}, with
meets given by intersection \cite[\S O-1.15, O-2.8]{gierz_2003}. The bottom
element is $\{ \top \}$, while the binary join is $F_1 \lor F_2 =
\UpSet{\SetComp{ a \land b }{a \in F_1, b \in F_2}}$ \cite[\S
O-1.15]{gierz_2003}. Infinitary joins $\Sup F_i$ are given by $\UpSet{\SetComp{
a_{n_1} \land \ldots \land a_{n_j} }{ a_{n_k} \in F_{n_k} }}$. In fact, as $W$
is distributive, infinite joins and finite meets satisfy the \emph{infinite
distributive law}, making $\FUNC{W}{\TV}[\land]$ a \emph{locale}, or
\emph{complete Heyting algebra} (cHA) \cite[\S II.2.11]{johnstone_1982}. The
exponential is given by $\HeyExp{F_1}{F_2} = \SetComp{ w \in W } { \forall w
\sqsubseteq v.\ v \in F_1 \text{ implies } v \in F_2}$, which one can readily
check is a filter whenever $F_1$ and $F_2$ are---as long as $W$ is distributive.

Third, given any $w \in W$ its \emph{principal filter} $\UpSet{w}$ is
$\SetComp{v \in W}{w \sqsubseteq v} \in \FUNC{W}{\TV}[\land]$. As $w \sqsubseteq
v$ iff $\UpSet{v} \subseteq \UpSet{w}$, this gives an \emph{order-embedding} $
\UpSet : \Op{W} \to \FUNC{W}{\TV}[\land]$. The key to this paper is the
following lemma.
\begin{lemma}
  \label{lemma:upper-preserves}
  $\UpSet : \Op{W} \to \FUNC{W}{\TV}[\land]$ preserves finite and infinite
  meets, \textbf{finite joins}, and exponentials.
\end{lemma}
(The dual of) most of this lemma can be found in \cite[\S O-2.15]{gierz_2003};
the rest is elementary---at least if one notices that the domain of $\UpSet$ is
the \emph{opposite} of $W$.

Fourth, the principal upper sets $\UpSet{w}$ are special, as they are
\emph{compact}. Let $L$ be a dcpo. An element $d \in L$ is compact just if $d
\sqsubseteq \bigsqcup^{\uparrow} X$ implies that $d \sqsubseteq x$ for some $x
\in X$, for any directed set $X$. Like (completely) prime elements, this says
that $d$ contains a small, indivisible fragment of information: as soon as it
approximates a \emph{directed} supremum, i.e.~a `recursively defined element,'
it must approximate some `finite unfolding' of it. We write $\Compact{L}$ for
the set of compact elements of $L$. It is not hard to show that the compact
elements of $\FUNC{W}{\TV}[\land]$ are exactly the principal upper sets
$\UpSet{w}$ for some $w \in W$ \cite[\S I-4.10]{gierz_2003} \cite[Prop.
2.2.2]{abramsky_1994}. Due to the finitary cases of Lemma
\ref{lemma:upper-preserves} this means that the sub-poset of compact elements
$\Compact{\FUNC{W}{\TV}[\land]}$ is in fact a \emph{sub-lattice} of
$\FUNC{W}{\TV}[\land]$. This fact will prove very important.

Fifth, the complete lattice $\FUNC{W}{\TV}[\land]$ is \emph{algebraic} \cite[\S
I-4.10]{gierz_2003} \cite[\S 2.2]{abramsky_1994}. This means that all its
elements can be reconstructed as directed suprema of compact ones. In symbols,
$L$ is algebraic just if 
\[
  \text{for any $d \in L$, } \textsf{K}_d \defeq \SetComp{c \in \Compact{L}}{c
  \sqsubseteq d} \text{ is directed and } d = \DirectedSup \textsf{K}_d
\]
In summary, if $W$ is distributive, $\FUNC{W}{\TV}[\land]$ is a frame which is
(i) algebraic, and (ii) whose compact elements form a sub-lattice. Such lattices
are referred to as \emph{spectral locales} (or \emph{coherent frames}), and play
an important r\^{o}le in Stone duality. In fact, every such locale arises as the
partial order of filters of a distributive lattice \cite[\S
II.3.2]{johnstone_1982}:
\begin{theorem}
  \label{theorem:johnstone}
  A frame is coherent iff it is isomorphic to $\FUNC{W}{\TV}[\land]$ for a
  distributive lattice $W$.
\end{theorem}
\noindent This theorem says that a coherent frame $L$ is isomorphic to the
filters $\Filt*{\Compact{L}}$ of its compact elements.


Finally, the fact that every element can be reconstructed as a supremum of
compact elements means that it is possible to canonically extend any monotonic
$f : W \to W'$ that \emph{preserves finite joins} to a monotonic,
join-preserving $\FUNC{\Op{W}}{\TV} \to W'$, as long as $W'$ is a complete
lattice. Diagrammatically, in the situation
\begin{equation}
  \label{diagram:cont-extension-poset}
  \begin{tikzpicture}[node distance=2.5cm, on grid, baseline=(current  bounding  box.center)]
    \node (C) {$W$};
    \node (PSHC) [right = 4cm of C] {$\FUNC{\Op{W}}{\TV}[\land]$};
    \node (E) [below = 2cm of PSHC] {$W'$};
    \path[->] (C) edge node [above] {$\UpSet$} (PSHC);
    \path[->] (C) edge node [below = .5em] {$f$} (E);
    \path[->, dashed] (PSHC) 
      edge 
        node [left] {$\LKan{f}$} 
        node [right = .15em] {$\Adjoint$}
      (E);
    \path[->, bend right = 45, dotted] (E) edge node [right] {$\NerveS{f}$} (PSHC);
  \end{tikzpicture}
\end{equation}
there exists a unique $\LKan{f}$ which preserves all joins and satisfies
$\LKan{f}\prn{\UpSet w} = f(w)$. It is given by
\[
  \LKan{f}\prn{S} \defeq \Sup \SetComp{ f(w) }{ w \in S }
\]
We call $\LKan{f}$ the \emph{Scott-continuous extension} of $f$ along $\UpSet$.
This follows from a much more general theorem: if $W$ is merely a poset and $W'$
a dcpo, then $\FUNC{\Op{W}}{\TV}[\land]$ is a dcpo, and there is a unique
Scott-continuous $\LKan{f}$ that makes \eqref{diagram:cont-extension-poset}
commute \cite[Prop. 2.2.24]{abramsky_1994}. However, if $W$ already has finite
joins, then $\FUNC{\Op{W}}{\TV}$ is a complete lattice. The reason is every join
can be written as  a directed supremum of non-empty finite ones. Then, if $f$
preserves finite joins, $\LKan{f}$ preserves all of them. 
As $\FUNC{\Op{W}}{\TV}[\land]$ is complete, it has a right adjoint $\NerveS{f}$,
by the adjoint functor theorem \cite[\S 7.34]{davey_2002} \cite[\S
I.4.2]{johnstone_1982}.

Suppose then that we start with a stable model $(W, \sqsubseteq, V)$. By taking
its filters we then obtain a spectral locale $\FUNC{W}{\TV}[\land]$. Defining
$\sem{p} = V(p)$ we obtain a model of intuitionistic logic which interprets
every formula $\varphi$ as a filter $\sem{\varphi} \in \FUNC{W}{\TV}[\land]$,
namely the filter of worlds in which it is true:
\begin{proposition}
  \label{proposition:stable-vs-algebraic-1}
  $\KSat{w}{\varphi}$ if and only if $w \in \sem{\varphi}$.
\end{proposition}
%

In view of this proposition,
\begin{theorem}[Soundness]
  The stable semantics is sound for intuitionistic logic.
\end{theorem}

\subsection{Completeness}
\label{section:int-completeness}

Revisiting the remarks of \S\ref{section:int-models-as-frames} we can prove
that completeness of the stable semantics implies completeness of the algebraic
semantics and vice versa. One direction works exactly as it would for Kripke
semantics:
\begin{theorem}
  \label{theorem:stable-to-alg}
  Completeness of the stable semantics implies completeness of the algebraic
  semantics.
\end{theorem}
\begin{proof}
  Suppose $\sem{\varphi}_H = 1$ for every Heyting algebra $H$ and any
  interpretation $\sem{p}_H \in H$ of the propositions. Hence, given any stable
  model $(W, \sqsubseteq, V)$ we have that $\sem{\varphi}_{\FUNC{W}{\TV}[\land]}
  = 1 = W$ where $\sem{p}_{\FUNC{W}{\TV}[\land]} = V(p)$. But Prop.
  \ref{proposition:stable-vs-algebraic-1} then implies that $\KSat{w}{\varphi}$ for
  all $w \in W$. By completeness of the stable semantics, $\vdash \varphi$.
\end{proof}

However, the filter construction enables a proof of the other direction as well.
We have an embedding
\[
  \UpSet[\Op{H}] : H \to \FUNC{\Op{H}}{\TV}[\land]
\]
of any Heyting algebra $H$ into the cHA of filters of $\Op{H}$ which is a
Heyting homomorphism by Lemma \ref{lemma:upper-preserves}. It is worth pausing
for a moment to ponder that a filter on $\Op{H}$ is in fact an \emph{ideal} of
$H$, viz. a lower set that is a sub-join-semilattice. Thus, $\UpSet[\Op{H}]$
sends $x \in H$ to the \emph{principal ideal} $\SetComp{ y \in H }{ y \leq x }$
of $x$ in $H$.

Suppose we have a Heyting algebra $H$, and some interpretation of propositional
variables $\sem{p}_{H} \in H$. Define an interpretation into
$\FUNC{\Op{H}}{\TV}[\land]$ starting from $\sem{p}_{\FUNC{\Op{H}}{\TV}[\land]}
\defeq \UpSet*[\Op{H}]{\sem{p}_{H}}$. Then, by Lemma
\ref{lemma:upper-preserves},
\begin{proposition}
  \label{proposition:upset-filters}
  $
  \sem{\varphi}_{\FUNC{\Op{H}}{\TV}[\land]}
  =
  \UpSet[\Op{H}]{\sem{\varphi}_{H}}
  =
  \SetComp{ y \in H }{ y \leq \sem{\varphi}_{H}}
  $
\end{proposition}
We are now in a position to prove that
\begin{theorem}
  \label{theorem:alg-to-stable}
  Completeness of the algebraic semantics implies completeness of the stable
  semantics.
\end{theorem}
\begin{proof}
  Suppose $\varphi$ is valid in every stable model. By completeness of the
  algebraic semantics, it suffices to show that $\sem{\varphi}_{H} = 1$ for any
  Heyting algebra $H$, and any interpretation $\sem{p}_H \in H$, as this implies
  $\vdash \varphi$. Consider then the stable model $(\Op{H}, \sqsubseteq, V)$
  where $V(p) = \UpSet*[\Op{H}]{\sem{p}_H}$. As $\varphi$ is valid in this
  model, $\KSat{x}{\varphi}$ for every $x \in H$. By Proposition
  \ref{proposition:stable-vs-algebraic-1} and Proposition
  \ref{proposition:upset-filters} we get that $H =
  \sem{\varphi}_{\FUNC{\Op{H}}{\TV}[\land]} = \SetComp{y \in H}{y \leq
  \sem{\varphi}_{H}}$. 
\end{proof}

\subsection{Morphisms}
  \label{section:stable-morphisms}

We briefly consider what it means to have a morphism $f : W \to W'$ of stable
frames. We would like such morphisms to induce a map $\Pre{f} :
\FUNC{W'}{\TV}[\land] \to \FUNC{W}{\TV}[\land]$ by $\Pre{f}(F) = \SetComp{ v \in
W'}{f(v) \in F}$. To conclude that $\Pre{f}(F)$ is a filter we need to know that
$f$ is monotonic, and that it preserves finite meets. Such maps warrant their
own name, which we borrow from the literature on stable domain theory
\cite{berry_1978}:
\begin{definition}
  A monotonic map $f : W \to W'$ is \emph{stable} just if it preserves finite
  meets. We define $\STABLE$ to be the category of distributive lattices and
  stable maps.
\end{definition}
\noindent Unlike the category $\DLATT$ of distributive lattices, the morphisms
of $\STABLE$ need not preserve disjunctions.

It is straightforward to show that when $f$ preserves finite meets, $\Pre{f}$ is
Scott-continuous and preserves arbitrary meets. This defines a functor
$\FUNC{-}{\TV}[\land] : \Op{\STABLE} \fto \COH$, where $\COH$ is the category of
coherent frames and Scott-continuous, meet-preserving morphisms. Note that this
is \emph{not} the usual category that is used in Stone duality, whose morphisms
are frame maps that preserve compact elements \cite[\S II.3.3]{johnstone_1982}.

It is not difficult to see that $\FUNC{-}{\TV}[\land]$ is an equivalence. On
objects this is guaranteed by Theorem \ref{theorem:johnstone}. On morphisms, it
suffices to spot that every $\Pre{f} : \FUNC{W'}{\TV}[\land] \to
\FUNC{W}{\TV}[\land]$ preserves meets, and hence has a left adjoint $\LKan{f}$
by the adjoint functor theorem. It is then simple to show that left adjoints
preserve compact elements, so that $f$ can be extracted by restricting
$\LKan{f}$ to $\Compact{\FUNC{W}{\TV}[\land]} \cong W$. This leads to a duality
\begin{equation}
  \label{equation:stable-vs-coh}
  \Op{\STABLE} \Equiv \COH
\end{equation}
Weaker versions of this duality are well-known, see e.g. \cite[\S
IV-1.16]{gierz_2003} for a duality between meet-semilattices and algebraic
lattices, as well as references to it in the literature.

However, stable morphisms do not preserve truth. For that, we need to refine the
above duality to maps that are stable, open, surjective, and also
\emph{L-morphisms} in the sense of Bezhanishvili et al. \cite[\S
2]{bezhanishvili_2024}, which appropriately preserve disjunction. The details
are similar to those in my previous article \cite{kavvos_2024a}.

\section{Stable semantics of modal logic}
  \label{section:stable-modal}

I have previously argued that a canonical Kripke semantics for intuitionistic
modal logic is given by a \emph{bimodule}, i.e. a monotonic function $R : \Op{W}
\times W \to \TV$ over a Kripke frame $(W, \sqsubseteq)$ \cite{kavvos_2024a}. In
this section I will adapt this to the case where $(W, \sqsubseteq)$ is a stable
frame.
\begin{definition}
  A \emph{stable bimodule on $W$} is a bimodule $R : \Op{W} \times W \to \TV$
  that additionally satisfies the following \emph{stability conditions}:
  \begin{enumerate}
    \item
      $
          w \mathbin{R} v_1\,
          \text{ and }\,
          w \mathbin{R} v_2
        \Longrightarrow\
          w \mathbin{R} (v_1 \land v_2)
      $
    \item
      $w \mathbin{R} 1$
    \item
      $
        (w_1 \land w_2) \mathbin{R} v\
          \Longrightarrow
        \exists v_1, v_2.\ v_1 \land v_2 \sqsubseteq v
          \text{ and } w_1 \mathbin{R} v_1, w_2 \mathbin{R} v_2
      $
    \item
      $
        1 \mathbin{R} v\ \Longleftrightarrow v = 1
      $
  \end{enumerate}
\end{definition}
A bimodule is a relation $R \subseteq W \times W$ with the property that
$w' \sqsubseteq w \mathbin{R} v \sqsubseteq v'$ implies $w' \mathbin{R} v'$.
This automatically implies the converses of (i) and (iii). Condition (ii) is
redundant, as it is implied by (iv) and the bimodule conditions, but we keep it
for symmetry. A \emph{modal stable frame} $(W, \sqsubseteq, R)$ comprises a
stable frame $(W, \sqsubseteq)$ and a stable bimodule $R$.

Stability conditions (i) and (ii) ensure that abstracting the second variable
yields a monotonic map $\Lambda R : \Op{W} \to \FUNC{W}{\TV}[\land]$. Moreover,
stability conditions (iii) and (iv) ensure that $\Lambda R$ preserves finite
joins. Then $\lambda R$ induces the following adjunction by Scott-continuous
extension:
\begin{equation}
  \begin{tikzpicture}[node distance=2.5cm, on grid, baseline=(current  bounding  box.center)]
    \node (C) {$\Op{W}$};
    \node (PSHC) [right = 4cm of C] {$\FUNC{W}{\TV}[\land]$};
    \node (E) [below = 2cm of PSHC] {$\FUNC{W}{\TV}[\land]$};
    \path[->] (C) edge node [above] {$\UpSet$} (PSHC);
    \path[->] (C) edge node [below = .5em] {$\lambda R$} (E);
    \path[->, dashed] (PSHC) 
      edge 
        node [left] {$\blacklozenge_R$} 
        node [right = .15em] {$\Adjoint$}
      (E);
    \path[->, bend right = 45, dotted] (E) edge node [right] {$\Box_R$} (PSHC);
  \end{tikzpicture}
\end{equation}
Like in the case of Kripke semantics \cite{kavvos_2024a} it can be shown that
these maps are given by
\begin{align*}
  \blacklozenge_R(F)
    &\defeq \SetComp{ w \in W }{ \exists v.\ v \mathbin{R} w \text{ and } v \in
    F} &
  \Box_R(F)
    &\defeq \SetComp{ w \in W }{ \forall v.\ w \mathbin{R} v \text{ implies } v \in F }
\end{align*}
It is easy to show that both $\blacklozenge_R(F)$ and $\Box_R(F)$ are filters
whenever $F$ is; the proof of the first uses stability conditions (i) and (iv),
and the second uses stability conditions (iii) and (iv).

This directly leads to the following clauses of a stable semantics of the two
modalities:
\begin{align*}
  \KSat[\Model{M}]{w}{\blacklozenge \varphi}\
  &\defequiv\
    \exists v.\ v \mathbin{R} w \text{ and } \KSat[\Model{M}]{v}{\varphi}
  &
  \KSat[\Model{M}]{w}{\Box \varphi}\
  &\defequiv\
    \forall v.\ w \mathbin{R} v \text{ implies } \KSat[\Model{M}]{v}{\varphi}
\end{align*}
to which Proposition \ref{proposition:stable-vs-algebraic-1} readily extends. I
have neglected to mention what a \emph{modal stable model} $\Model{M} = (W,
\sqsubseteq, R, V)$ is: $(W, \sqsubseteq, R)$ is a modal stable frame, and
the valuation $V$ maps propositions into filters.

\subsection{Completeness}
\label{section:modal-completeness}

I have also previously argued that applying Kan extension to a bimodule
inescapably leads us to an intuitionistic modal logic with two adjoint
modalities, $\blacklozenge$ and $\Box$, as studied by Dzik et al.
\cite{dzik_2010,kavvos_2024a}. The two clauses of the stable semantics are
identical to the Kripke semantics in \emph{loc. cit}. But is the logic the same?
To answer that we have to reach for its algebraic models, which are Heyting
algebras $H$ equipped with two operators $\blacklozenge, \Box : H \to H$ that
form an adjunction $\blacklozenge \Adjoint \Box$. We have just seen that stable
bimodules on $W$ correspond precisely to such adjunctions on the cHA
$\FUNC{W}{\TV}[\land]$. As Proposition \ref{proposition:stable-vs-algebraic-1}
remains true if we include $\blacklozenge$ and $\Box$, we have that the stable
semantics is sound for the logic of Dzik et al. Furthermore,

\begin{theorem}
  Completeness of the modal stable semantics implies completeness of the modal
  algebraic semantics.
\end{theorem}
The proof is the same as that of Theorem \ref{theorem:stable-to-alg}. For the
other direction we have to combine our work from intuitionistic logic, and the
argument from \S\ref{section:int-models-as-frames}. Given a Heyting algebra $H$
and an adjunction on it, the map $\UpSet \circ \mathop{\blacklozenge}$ preserves
finite joins, as both $\UpSet$ and $\blacklozenge$ do. Take its Scott-continuous
extension, $\blacklozenge_f$: 
\begin{equation}
  \label{diagram:continuous-modal-scott}
  \begin{tikzpicture}[node distance=2.5cm, on grid, baseline=(current  bounding  box.center)]
    \node (C) {$H$};
    \node (PSHC) [right = 6cm of C] {$\FUNC{\Op{H}}{\TV}[\land]$};
    \node (E) [below = 2cm of PSHC] {$\FUNC{\Op{H}}{\TV}[\land]$};
    \node (C2) [below = 2cm of C] {$H$};
    \path[->] (C) edge node [above] {$\UpSet$} (PSHC);
    \path[->] (C2) edge node [below] {$\UpSet$} (E);
    \path[->] (C) 
      edge 
        node [left] {$\blacklozenge$} 
        node [right = .05em] {$\Adjoint$}
      (C2);
    \path[->, bend right = 45] (C2) edge node [right] {$\Box$} (C);
    \path[->, dashed] (PSHC) 
      edge 
      node [left] {$\blacklozenge_f$} 
        node [right = .15em] {$\Adjoint$}
      (E);
      \path[->, bend right = 45, dotted] (E) edge node [right] {$\Box_f$} (PSHC);
  \end{tikzpicture}
\end{equation}
The map $\UpSet \circ \mathop{\blacklozenge}$ corresponds to a stable
$R_\blacklozenge : H \times \Op{H} \to \TV$, which maps $(x, y)$ to $1$ iff $y
\leq \blacklozenge x$ in $H$. This is by definition stable, but in any case that
is easy to verify manually---as long as one is careful about variance. For
example, to prove (iii), we need to show that whenever $y \leq \blacklozenge x_1
\lor \blacklozenge x_2$ there exist $y_1, y_2$ with $y \leq y_1 \lor y_2$ and
$y_1 \leq \blacklozenge x_1$ and $y_2 \leq \blacklozenge x_2$. It suffices to
take $y_i = y \land \blacklozenge x_i$ and use distributivity.

Diagram \eqref{diagram:continuous-modal-scott} commutes. For $\blacklozenge$ we
have that $\UpSet \circ \mathop{\blacklozenge} = \mathop{\blacklozenge_f} \circ
\UpSet$ by definition. For the $\Box$ we have
\[
  \Box_f\prn{\UpSet z}
  = \SetComp{ x \in H }{ \forall y.\, y \leq \blacklozenge x \text{ implies } y \leq z }
  = \SetComp{ x \in H }{ \blacklozenge x \leq z }
  = \SetComp{ x \in H }{ x \leq \Box z }
  = \UpSet{\Box z}
\]
Proposition \ref{proposition:upset-filters} extends to the modal case. We
therefore have
\begin{theorem}
  Completeness of the modal algebraic semantics implies completeness of the
  modal stable semantics.
\end{theorem}
The proof is also like that of Theorem \ref{theorem:alg-to-stable}. Thus, the
modal stable semantics is sound and complete for the intuitionistic modal logic
of Dzik et al. \cite{dzik_2010}.

\subsection{Morphisms}
  \label{section:modal-morphisms}

The duality \eqref{equation:stable-vs-coh} of \S\ref{section:stable-morphisms}
can be restricted to a duality
\[
  \textbf{SBimod}^\text{op} \Equiv \textbf{CohO}
\]
The category on the left has distributive lattices equipped with a stable
bimodule as objects; and stable morphisms that preserve the bimodule as
morphisms. The category on the right has coherent frames $L$ equipped with a
meet-preserving operation $\Box_L : L \to L$ as objects; and Scott-continuous,
meet-preserving maps $h : L \to L'$ for which $h \Box_L \sqsubseteq \Box_{L'}
h$. In analogy with previous results this can be further refined to a duality
where the morphisms preserve truth on the left, and the operator and implication
on the right.

\section{Two-dimensional stable semantics of intuitionistic logic}
  \label{section:2dss-int}

Following the programme of two-dimensional Kripke semantics\cite{kavvos_2024a},
we look for categorifications of the stable semantics. Thus, we exchange stable
frames $(W, \sqsubseteq)$ for arbitrary categories $\CC$ with finite products;
we could call these \emph{stable categories}. The first thing we must categorify
is the notion of \emph{filter}. Surprisingly, there are two possible choices:
\begin{enumerate}
  \item the \emph{Ind-completion} $\IND{\CC}$, which adds all \emph{filtered
    colimits} to $\CC$ \cite[\S VI.1]{johnstone_1982} \cite[\S
    4.17]{adamek_2010}; and
  \item the \emph{Sind-completion} $\SIND{\CC}$, which adds all \emph{sifted
    colimits} to $\CC$ \cite{adamek_2001,adamek_2010b,adamek_2010}.
\end{enumerate}
All filtered colimits are sifted, so the latter involves adding strictly more
colimits to $\CC$. However, for a poset $W$ we have that $\IND{\Op{W}} \cong
\SIND{\Op{W}} \cong \Filt*{W}$ \cite[\S 2.3]{adamek_2001}. Thus, these two
completions are \emph{indistinguishable} at the order-theoretic level. As an
aside, note that the former completion is related to \emph{essentially algebraic
theories} \cite{adamek_1994}, while the latter to simpler \emph{algebraic
theories} of Lawvere \cite{adamek_2010}. 

We will work with the sifted completion, for several reasons. The most important
one is that, when $\CC$ has finite coproducts, $\SIND{\CC}$ is cocomplete. This
is just enough to allow us to embed any bicartesian closed category $\CC$ (which
has coproducts) into a cocomplete category $\SIND{\CC}$, adapting the story of
\S\ref{section:stable}. The cocompleteness is absolutely essential in the
semantics of modalities given in \S\ref{section:2dss-modal}. Second, the
conditions required on $\CC$ for $\SIND{\CC}$ to be a cartesian closed
category---and hence a model of intuitionistic proofs---are rather weak. Third,
there is an analogy to working with filters as elements of
$\FUNC{W}{\TV}[\land]$: the classic Lawverean move of replacing $\TV$ by $\SET$
\cite{lawvere_1973} leads us to consider product-preserving presheaves
$\FUNC{\CC}{\SET}[\times]$, which coincide with $\SIND{\Op{\CC}}$ whenever $\CC$
has products, mirroring Proposition \ref{proposition:filters-meet-semilattice}.


The following proposition collects various well-known facts about the sifted
completion \cite{adamek_2001,adamek_2010}. These are analogous to well-known
facts about presheaf categories \cite{kavvos_2024a}, and mirror those of
$\Filt*{W}$ given in \S\ref{section:stable}.
\begin{proposition}
  \label{proposition:sifted}
  Let $\SIND{\CC}$ be the sifted completion of $\CC$.
  \begin{enumerate}
    \item If $\CC$ has coproducts then $\SIND{\CC}$ is isomorphic to the
      category $\FUNC{\Op{\CC}}{\SET}[\times]$ of \emph{product-preserving
      presheaves} and natural transformations. It is a complete and cocomplete
      category.
  \end{enumerate}
  For the rest of this proposition we assume that $\CC$ has coproducts.
  \begin{enumerate} \setcounter{enumi}{1}
    \item Representable presheaves are product-preserving, so $\Yo : \CC \fto
      \FUNC{\Op{\CC}}{\SET}[\times]$ is an embedding.
    \item $\Yo : \CC \fto \FUNC{\Op{\CC}}{\SET}[\times]$ preserves products and
      \textbf{coproducts}.
    \item $\Yo{c}$ is \emph{perfectly presentable}, i.e. $\Hom{\Yo{c}}{-} :
      \FUNC{\Op{\CC}}{\SET}[\times] \to \SET$ preserves sifted colimits.
    \item A category is equivalent to $\FUNC{\Op{\CC}}{\SET}[\times]$ for some
      $\CC$ if and only if it is cocomplete and has a strong generator
      consisting of perfectly presentable objects. Moreover, there is a unique
      idempotent-complete category $\CC$ for which this is true (up to
      equivalence): the subcategory of perfectly presentable objects.
  \end{enumerate}
\end{proposition}
\begin{proof}
  (i) is shown in \cite[\S 2.8]{adamek_2001} \cite[\S 1.22, 4.5,
  4.13]{adamek_2010}.
  (ii) is shown in \cite[\S 1.12]{adamek_2010} and (iii) in \cite[\S
  1.13]{adamek_2010}. (iv) follows from the fact representables are \emph{tiny},
  and that $\FUNC{\Op{\CC}}{\SET}[\times]$ is closed under sifted colimits in
  presheaves \cite[\S 5.5]{adamek_2010}. (v) is shown in \cite{adamek_2001}
  \cite[\S 6.9, 8.12]{adamek_2010}. 
\end{proof}

Categories satisfying the fifth condition given above are called \emph{algebraic
categories} by Lawvere \cite{lawvere_1963}. They are essentially categories of
models of algebraic theories:  see the textbook by Adamek et al.
\cite{adamek_2010}.

One might wonder whether $\Yo : \CC \fto \FUNC{\Op{\CC}}{\SET}[\times]$
preserves exponentials. It would, were $\FUNC{\Op{\CC}}{\SET}[\times]$ to have
them; and it has them exactly when $\CC$ is a \emph{distributive category}, i.e.
whenever the canonical morphism $(a \times b) + (a \times c) \to a \times (b +
c)$ is an isomorphism \cite{cockett_1993}:
\begin{proposition}[Fiore]
  \label{proposition:sifted-ccc}
  Let $\CC$ have (co)products. Then, the following are equivalent:
  \begin{enumerate}
    \item $\CC$ is distributive.
    \item $P \times \Yo{a + b} \cong P \times \Yo{a} + P \times \Yo{b}$ in
      $\FUNC{\Op{\CC}}{\SET}[\times]$
    \item $\SIND{\CC} \cong \FUNC{\Op{\CC}}{\SET}[\times]$ is cartesian closed.
  \end{enumerate}
  In that case $\Yo : \CC \fto \FUNC{\Op{\CC}}{\SET}[\times]$ preserves
  exponentials.
\end{proposition}
\begin{proof}
  (i) $\Rightarrow$ (ii): Write $P \cong \Colim[(c, x) \in \PSHEL{P}] \Yo{c}$ as
  a colimit of representables. As $P$ is product-preserving, its category of
  elements $\PSHEL{P}$ is sifted \cite[\S 4.2]{adamek_2010}. Hence, it does not
  matter if this colimit is in presheaves or product-preserving presheaves, as
  the latter are closed under sifted colimits within the former \cite[\S
  2.5]{adamek_2010}. Noticing also that $\times$ is the same operation in both
  $\FUNC{\Op{\CC}}{\SET}$ and $\FUNC{\Op{\CC}}{\SET}[\times]$ we may calculate
  \begin{align*}
    P \times \Yo{a + b}
    &\cong \prn{\Colim[(c, x) \in \PSHEL{P}] \Yo{c}} \times \Yo{a + b} && \text{now in presheaves}\\
    &\cong \Colim[(c, x) \in \PSHEL{P}] \prn{\Yo{c} \times \Yo{a + b}} && \text{as $- \times \Yo{a + b}$ is left adjoint} \\
    &\cong \Colim[(c, x) \in \PSHEL{P}] \Yo{c \times (a + b)} && \text{now back in $\FUNC{\Op{\CC}}{\SET}[\times]$} \\
    &\cong \Colim[(c, x) \in \PSHEL{P}] \Yo{(c \times a) + (c \times b)} && \\
    &\cong \Colim[(c, x) \in \PSHEL{P}] (\Yo{c} \times \Yo{a}) + (\Yo{c} \times \Yo{b})  && \text{where $+$ is now in $\FUNC{\Op{\CC}}{\SET}[\times]$} \\
    &\cong 
      \prn{\Colim[(c, x) \in \PSHEL{P}] \Yo{c} \times \Yo{a}}
      +
      \prn{\Colim[(c, x) \in \PSHEL{P}] \Yo{c} \times \Yo{b}} && \text{as colimits commute with colimits} \\
    &\cong 
      \prn{\Colim[(c, x) \in \PSHEL{P}] \Yo{c}} \times \Yo{a}
        +
      \prn{\Colim[(c, x) \in \PSHEL{P}] \Yo{c}} \times \Yo{b} \\
    &\cong P \times \Yo{a} + P \times \Yo{b}
  \end{align*}

  (ii) $\Rightarrow$ (iii): We only need to prove that the usual exponential
  $
  (\HeyExp{P}{Q})(c) \defeq \Hom[\FUNC{\Op{\CC}}{\SET}[\times]]{P \times \Yo{c}}{Q}
  $
  is a product-preserving presheaf. But this easily follows from the observation
  that $\Yo{\initial} \cong \terminal$ and (ii).

  (iii) $\Rightarrow$ (i): Then $\FUNC{\Op{\CC}}{\SET}[\times]$ is a bicartesian
  closed category, and hence it is distributive. But $\Yo$ is an embedding that
  preserves products and coproducts, so the subcategory $\CC$ is distributive as
  well.
\end{proof}

This result is a special case of a more general theorem due to Marcelo Fiore
\cite[Prop 11.10]{fiore_1996} \cite[Prop 2.4]{fiore_1997} \cite[\S
3.1]{fiore_2002} \cite[\S 6.2]{curien_2016}. The proof we give here simplifies
one given by Younesse Kaddar \cite{kaddar_2020}.

This puts us in a good place to introduce a \emph{two-dimensional stable
semantics}. This amounts to replacing the stable frame $(W, \sqsubseteq)$ by a
category $\CC$ with products and coproducts for which $\Op{\CC}$ is
distributive. This means that the somewhat unusual isomorphism $a + (b \times c)
\cong (a + b) \times (a + c)$ holds in $\CC$. This is known to hold in a number
of categories of algebras, including distributive lattices and commutative rings
\cite{davey_1981}. I have yet to grasp the meaning of this for a general Lawvere
theory. Perhaps previous work by Johnstone \cite{johnstone_1990} and Garner
\cite{garner_2023} describing when categories of varieties are toposes or
cartesian closed might give some insight. Finally, note that---unlike
distributive lattices---distributive categories are \emph{not self-dual}, so
this is all we get.

By Propositions \ref{proposition:sifted} and \ref{proposition:sifted-ccc}, the
category $\FUNC{\CC}{\SET}[\times]$ is a bicartesian closed category. The
two-dimensional stable semantics is then dictated by the bicartesian closed
structure. Thus, every formula $\varphi$ is interpreted as a product-preserving
presheaf
\[
  \sem{\varphi} : \CC \to_{\times} \SET
\]
Writing $\sem{\varphi}_w$ to mean $\sem{\varphi}(w)$ for any $w \in \CC$ and $f
\cdot x \in \sem{\varphi}_{v}$ for $\sem{\varphi}(f)$ and $f : w \to v$, the
clauses are the expected proof-relevant categorifications of the stable
semantics of \S\ref{section:stable}. Some are as before:
\begin{align*}
    \sem{\varphi \land \psi}_w &\defeq \sem{\varphi}_w \times \sem{\psi}_w \\
    \sem{\varphi \to \psi}_w &\defeq 
    \SetComp{ 
      F : \DepFun{v}{\CC}{\Hom[\CC]{w}{v} \to \sem{\varphi}_v \to \sem{\psi}_v}
    }{ \forall g.\,\, 
      g \cdot F(v)(f)(x) = F(v')(g \circ f)(g \cdot x)  }
\end{align*}

However, the interpretation of $\varphi \lor \psi$ is not immediately evident,
as it is a coproduct in $\FUNC{\CC}{\SET}[\times]$. Adamek et al. \cite[\S
4.5]{adamek_2010} prove the existence of such coproducts abstractly, by
decomposing presheaves as sifted colimits of representables and using the fact
$\Yo{a} + \Yo{b} = \Yo{a \times b}$. To enable a direct comparison with the
stable semantics of disjunction of \S\ref{section:stable}, we need to describe
them in a more direct way.

\begin{theorem}
  Let $\CC$ have finite products. Then the coproduct in
  $\FUNC{\CC}{\SET}[\times]$ is given by the coend
  \[
    (P + Q)(c) \defeq \int^{c_1, c_2 \in \CC} \Hom[\CC]{c_1 \times c_2}{c}
    \times P(c_1) \times Q(c_2)
  \]
\end{theorem}
An element of $(P + Q)(c)$ essentially consists of tuples $(c_1, c_2, f, x, y)$
where $f : c_1 \times c_2 \to c$ is a `decomposition' of $c$, $x \in P(c_1)$,
and $y \in Q(c_2)$. If we think of $\CC$ as an algebraic theory, $f$ can be
thought of as a term of sort $c$ in terms of two variables of sorts $c_1$ and
$c_2$; and the elements of $P(c_1)$ and $Q(c_2)$ can be considered as elements
of the algebra at sorts $c_1$ and $c_2$ respectively.

This is evidently a direct categorification of the stable semantics of
disjunction given in \S\ref{section:stable}. However, as this is now a coend,
these data have to be appropriately quotiented: for any $g : c'_1 \to c_1$, $h :
c'_2 \to c_2$, $t'_1 \in P(c'_1)$ and $t'_2 \in P(c'_2)$, the tuples $(c_1, c_2,
f, g \cdot t'_1, h \cdot t'_2)$ and $(c'_1, c'_2, f \circ (g \times h), t'_1,
t'_2)$ should be identified. This guarantees that the choice of decomposition is
`minimal.' It is easy to prove that this object has the right universal
property. However, a conceptual proof that it is product-preserving eludes me.

Finally, notice that this is essentially the `free' product of algebras, as
expected. It is also clearly a version of the \emph{Day convolution product} on
presheaves \cite[\S 6.2]{loregian_2021}. It is in fact a known result of higher
algebra that the Day convolution is the coproduct of commutative algebra objects
over a symmetric monoidal $\infty$-category: see Lurie's book \cite[Lemma
3.2.4.7]{lurie_2017}.

\subsection{Completeness}
  \label{section:2dss-int-completeness}

We are now able to show completeness results for the categorical semantics of
intuitionistic logic, i.e. bicartesian closed categories: if $\CC$ is a
bicartesian closed category then it is a fortiori distributive, and
$
  \Yo : \CC \fto \FUNC{\Op{\CC}}{\SET}[\times]
$
is an embedding that preserves the bicartesian closed structure of $\CC$. Lemma
\ref{lemma:ccc-conservativity-1} extends to disjunction and falsity on objects,
but also to \emph{proofs}. The latter can be represented as terms of the typed
$\lambda$-calculus with sums and an empty type up to $\beta\eta$ equality. We
refer to \cite{crole_1993,lambek_1988} for background on the categorical
semantics of the typed $\lambda$-calculus.
\begin{lemma}
  \label{lemma:ccc-conservativity-2}
  Let $\CC$ be a bicartesian closed category.
  \begin{enumerate}
    \item There is an isomorphism $\theta_A :
      \sem{A}_{\FUNC{\Op{\CC}}{\SET}[\times]} \cong \Yo{\sem{A}_{\CC}}$ for any
      type $A$ of the simply-typed $\lambda$-calculus.
    \item If $\Gamma \vdash M : A$ is a term of the typed $\lambda$-calculus,
      then the following diagram commutes:
      \[
        \begin{tikzpicture}[diagram]
          \SpliceDiagramSquare{
            height = 2cm,
            width = 4cm,
            nw = \sem{\Gamma}_{\FUNC{\Op{\CC}}{\SET}[\times]},
            sw = \Yo{\sem{\Gamma}_{\CC}},
            west = i \circ \prod_{(x : A) \in \Gamma} \theta_{A},
            ne = \sem{A}_{\FUNC{\Op{\CC}}{\SET}[\times]},
            se = \Yo{\sem{A}_{\CC}},
            east = \theta_A,
            south = \Yo{\sem{M}},
            north = \sem{M}
          }
        \end{tikzpicture}
      \]
      where $\sem{\Gamma} \defeq \prod_{(x : A) \in \Gamma} \sem{A}$, and $i :
      \prod_{(x : A) \in \Gamma} \Yo{\sem{A}_{\CC}} \xrightarrow{\cong}
      \Yo{\prod_{(x : A) \in \Gamma} \sem{A}_{\CC}}$ arises from the fact $\Yo$
      preserves finite products.
  \end{enumerate}
\end{lemma}
Then, assuming that bicartesian closed categories are complete for the typed
$\lambda$-calculus:
\begin{theorem}
  The subclass of models consisting of $\FUNC{\Op{\CC}}{\SET}[\times]$ over a
  distributive $\CC$ is complete for equational theory of the typed
  $\lambda$-calculus with sums and an empty type.
\end{theorem}
\begin{proof}
  Let $\Gamma \vdash M, N : A$ be two terms with $\sem{M} = \sem{N}$ when
  interpreted in any product-preserving presheaf category
  $\FUNC{\Op{\CC}}{\SET}[\times]$ with $\CC$ distributive. Pick any bicartesian
  closed $\CC$. By Lemma \ref{lemma:ccc-conservativity-2} we have $\Yo{\sem{M}}
  = \Yo{\sem{N}}$, where the interpretation is now in $\CC$. But $\Yo$ is
  faithful, so $\sem{M} = \sem{N}$ in every bicartesian closed category $\CC$.
  Then $\Gamma \vdash M = N : A$ by the completeness of bicartesian closed
  categories.
\end{proof}
There is of course a converse, which shows that completeness of this class of
models implies completeness of the class of bicartesian closed categories. It is
similar in spirit to Theorem \ref{theorem:stable-to-alg}.

\subsection{Morphisms}
  \label{section:2dss-int-morphisms}

Unlike most the previous dualities we have presented, the one in this section
has been carefully studied \cite{adamek_2003,adamek_2010}. However, the
terminology is different: instead of \emph{stable categories} they speak of
\emph{algebraic categories}; and instead of \emph{stable functors}, i.e.
functors preserving finite products, they speak of \emph{morphisms of algebraic
theories}. In fact, the duality required here is a \emph{syntax-semantics
duality} for Lawvere's algebraic theories.

To sketch this duality we must first look at the extension properties of
$\SIND{\CC}$. Given any $F : \CC \fto \EE$, where $\EE$ is a category with
sifted colimits, there is a unique $\LKan{F} : \FUNC{\Op{\CC}}{\SET}[\times]
\fto \EE$ that extends $F$ and preserves sifted colimits, as in the following
commuting diagram:
\begin{equation}
  \label{diagram:sift-extension-1}
  \begin{tikzpicture}[node distance=2.5cm, on grid, baseline=(current  bounding  box.center)]
    \node (C) {$\CC$};
    \node (PSHC) [right = 4cm of C] {$\FUNC{\Op{\CC}}{\SET}[\times]$};
    \node (E) [below = 2cm of PSHC] {$\EE$};
    \node (expl) [right = 5em of E] {has sifted colimits};
    \path[->] (C) edge node [above] {$\Yo$} (PSHC);
    \path[->] (C) edge node [below = .5em] {$F$} (E);
    \path[->, dashed] (PSHC) 
      edge 
      node [right] {$\LKan{F}$} 
      (E);
  \end{tikzpicture}
\end{equation}
This property is exactly what it means for $\FUNC{\Op{\CC}}{\SET}[\times]$ to be
the sifted colimit completion \cite[\S 4.9]{adamek_2010}. However, we can also
get a slightly more refined extension property. Suppose that $F : \CC \fto \EE$
also \textbf{preserves coproducts}, and that $\EE$ is cocomplete. Then
$\LKan{F}$ preserves \textbf{all colimits} and has a right adjoint:
\begin{equation}
  \label{diagram:sift-extension-2}
  \begin{tikzpicture}[node distance=2.5cm, on grid, baseline=(current  bounding  box.center)]
    \node (C) {$\CC$};
    \node (PSHC) [right = 4cm of C] {$\FUNC{\Op{\CC}}{\SET}[\times]$};
    \node (E) [below = 2cm of PSHC] {$\EE$};
    \node (expl) [right = 6em of E] {is cocomplete};
    \path[->] (C) edge node [above] {$\Yo$} (PSHC);
    \path[->] (C) edge node [below = .5em] {$F$} (E);
    \path[->, dashed] (PSHC) 
      edge 
        node [left] {$\LKan{F}$} 
        node [right = .15em] {$\Adjoint$}
      (E);
      \path[->, bend right = 45, dotted] (E) edge node [right] {$\NerveS{F}$} (PSHC);
  \end{tikzpicture}
\end{equation}
The reason is that the usual functor $\NerveS{F}(e) \defeq \Hom[\EE]{F-}{e}$ is
then valued in $\FUNC{\Op{\CC}}{\SET}[\times]$, and can readily be shown to be
right adjoint to $\LKan{F}$ \cite[\S 4.15]{adamek_2010}.

Then, given any stable functor $f : \CC \fto \DD$ take the extension of $\Yo
\circ \Op{f}$ as in \eqref{diagram:sift-extension-2}
\begin{equation}
  \label{diagram:sift-extension-stable}
  \begin{tikzpicture}[node distance=2.5cm, on grid, baseline=(current  bounding  box.center)]
    \node (C) {$\Op{\CC}$};
    \node (PSHC) [right = 6cm of C] {$\FUNC{\CC}{\SET}[\times]$};
    \node (E) [below = 2cm of PSHC] {$\FUNC{\CC}{\SET}[\times]$};
    \node (C2) [below = 2cm of C] {$\Op{\DD}$};
    \path[->] (C) edge node [above] {$\Yo$} (PSHC);
    \path[->] (C2) edge node [below] {$\Yo$} (E);
    \path[->] (C) 
      edge 
      node [left] {$\Op{f}$} 
      (C2);
    \path[->, dashed] (PSHC) 
      edge 
      node [left] {$\LKan{f}$} 
        node [right = .15em] {$\Adjoint$}
      (E);
      \path[->, bend right = 45, dotted] (E) edge node [right] {$\NerveS{f}$} (PSHC);
  \end{tikzpicture}
\end{equation}
We have that $\NerveS{f}(P) = \Hom{\Yo{f(-)}}{P} \cong P \circ f$ acts by
precomposition. It clearly preserves limits; it also preserves sifted colimits,
as they are computed pointwise in $\FUNC{\CC}{\SET}[\times]$ \cite[\S
9.3]{adamek_2010}. Such a functor is called an \emph{algebraic functor} \cite[\S
9.4]{adamek_2010}. We thus obtain a (strict) 2-functor
\[
  \FUNC{-}{\SET}[\times] : \textbf{Cat}_\text{cc, stable}^\text{op} \fto \ALGCAT
\]
from the (strict) 2-category of \emph{Cauchy-complete stable categories}, stable
functors, and natural transformations to the (strict) 2-category of
\emph{algebraic categories}, algebraic functors, and natural transformations.
This functor is a \emph{biequivalence}, and hence a \emph{2-duality}; more
details can be found in \cite[\S 9]{adamek_2010}. Finally, this can be further
refined to 2-dualities that `preserve truth' in terms of frames.

\subsection{Sheaves?}

Stable semantics superficially appears similar to the Kripke-Joyal semantics in
a sheaf category over an appropriate site. Indeed, suppose that $\CC$ is an
\emph{extensive} category \cite{carboni_1993}. As it also has finite products,
it is distributive \cite[\S 2]{cockett_1993}. We could then take the extensive
topology on $\CC$, where covers are given by the disjoint injections into finite
coproducts \cite[\S 4.6]{cockett_1993}. In that case the category
$\FUNC{\CC}{\SET}[\times]$ of product-preserving presheaves coincides with the
sheaf category over $\CC$ with the extensive topology, and the Kripke-Joyal
semantics seem rather close to the stable ones (this argument was pointed out to
me by Jonathan Sterling).

However, this does \emph{not} coincide with the stable semantics. The reason is
that the assumption that $\CC$ is extensive is too strong. In particular, no
partial order with non-trivial meets is extensive. Moreover, all sheaf toposes
have disjoint coproducts, and any finite-coproduct-preserving embedding $\CC
\fto \textsf{Sh}(\DD, J)$ reflects that disjointness back into $\CC$ (this
argument was pointed out to me by Fabian Ruch). But we are interested in many
non-extensive base categories, including all cartesian closed ones.

\section{Two-dimensional stable semantics for modal logic}
\label{section:2dss-modal}

\begin{definition}
  A \emph{stable profunctor} on $\CC$ is a profunctor $R : \Op{\CC} \times \CC
  \fto \SET$ which preserves products in its second argument, and for which
  $\Lambda R : \Op{\CC} \fto \FUNC{\CC}{\SET}[\times]$ preserves coproducts.
\end{definition}

This corresponds precisely to an adjunction on $\FUNC{\CC}{\SET}[\times]$, by
the universal property \eqref{diagram:sift-extension-2}:
\begin{equation*}
  \begin{tikzpicture}[node distance=2.5cm, on grid, baseline=(current  bounding  box.center)]
    \node (C) {$\Op{\CC}$};
    \node (PSHC) [right = 4cm of C] {$\FUNC{\CC}{\SET}[\times]$};
    \node (E) [below = 2cm of PSHC] {$\FUNC{\CC}{\SET}[\times]$};
    \path[->] (C) edge node [above] {$\Yo$} (PSHC);
    \path[->] (C) edge node [below = .5em] {$\Lambda R$} (E);
    \path[->, dashed] (PSHC) 
      edge 
        node [left] {$\blacklozenge_R$} 
        node [right = .15em] {$\Adjoint$}
      (E);
      \path[->, bend right = 45, dotted] (E) edge node [right] {$\Box_R$} (PSHC);
  \end{tikzpicture}
\end{equation*}
These functors are more directly expressed as follows:
\begin{align*}
  \label{equation:modal-2d}
  \sem{\blacklozenge \varphi}_w
  \defeq \blacklozenge_R \sem{\varphi}_w
  &= \int^{v \in \CC} \sem{\varphi}_v \times R(v, w)
  &
  \sem{\Box \varphi}_w
  \defeq \Box_R \sem{\varphi}_w
  &= \Hom[\FUNC{\CC}{\SET}[\times]]{R(w, -)}{\sem{\varphi}}
\end{align*}
The expression for $\Box$ follows from \eqref{diagram:sift-extension-2}; it
evidently preserves products. The expression for $\blacklozenge$ is the coend
formula for the left Kan extension along Yoneda to \emph{all presheaves}
\cite[\S 2.3]{loregian_2021}. It is still the right expression, by the
uniqueness of adjoints. However, it is not easy to see that it preserves
products in $w$: to see that, write $\sem{\varphi}$ as a sifted colimit and use
the rules of coends to show that this set is isomorphic to $\Colim[(v, x) \in
\PSHEL{\sem{\varphi}}] R(v, w)$. The latter clearly preserves products in $w$:
$R(v, -)$ does, and the colimit is sifted.

As in \cite{kavvos_2024a}, these are the expected categorifications of the
semantics of $\blacklozenge$ and $\Box$.

\subsection{Completeness}
  \label{section:2dss-modal-completeness}

We can now show another completeness result like that of
\S\ref{section:2dss-int}, which applies to \emph{intuitionistic modal proofs}.
These are bicartesian closed categories $\CC$ equipped with an adjunction
$\blacklozenge \Adjoint \Box$. They can be represented syntactically by
Clouston's \emph{Fitch-style $\lambda$-calculus} which is sound and complete for
such models \cite{clouston_2018}. Then
$
  \Yo : \CC \fto \FUNC{\Op{\CC}}{\SET}[\times]
$
is an embedding that preserves all this structure: Scott's lemma
\ref{lemma:ccc-conservativity-1} extends to $\blacklozenge$ and $\Box$,
following exactly the proof in \S\ref{section:int-models-as-frames}. Then, a
result similar to Lemma \ref{lemma:ccc-conservativity-2} holds, leading to the
\begin{theorem}
  The subclass of models consisting of categories
  $\FUNC{\Op{\CC}}{\SET}[\times]$ over a distributive $\CC$ equipped with an
  adjunction $\blacklozenge \Adjoint \Box$ on $\FUNC{\Op{\CC}}{\SET}[\times]$ is
  complete for equational theory of intuitionistic modal proofs.
\end{theorem}

\subsection{Morphisms}
  \label{section:2dss-modal-morphisms}

The 2-duality of \S\ref{section:2dss-int-morphisms} can be restricted to a
2-duality
\[
  \textbf{SProf}^\text{op}_\text{cc} \Equiv \textbf{AlgCatO}
\]
The (strict) 2-category on the left has Cauchy-complete categories equipped with
a stable profunctor as 0-cells; stable functors that preserve the profunctor as
1-cells; and natural transformations natural transformations. The (strict)
2-category on the right has algebraic categories $\mathcal{A}$ equipped with an
\emph{operation} $\Box_\mathcal{A} : \mathcal{A} \fto \mathcal{A}$ that
preserves limits and sifted colimits as 0-cells; algebraic functors $F :
\mathcal{A} \fto \mathcal{B} $ equipped with natural transformations
$F\Box_\mathcal{A} \To \Box_\mathcal{B} F$ as 1-cells; and natural
transformations as 2-cells. This is essentially a direct categorification of the
duality of \S\ref{section:modal-morphisms}. It can be further refined to a
2-duality where the morphisms preserve truth on the left, and the operator and
implication on the right.

\section{Related work}

\subsection*{Completions}

Much of the development in \S\ref{section:stable} was based on the \emph{filter
completion} of a distributive lattice. The dual notion of \emph{ideal
completion} is far more commonly encountered. It plays a significant r\^{o}le in
domain theory, as the ideal completion of a preorder is the \emph{free algebraic
dcpo} over an arbitrary set of compact-elements-to-be \cite[\S
2.2.6]{abramsky_1994}. The category of algebraic dcpos and continuous maps is
then equivalent to the category of preorders and \emph{approximable relations},
which appear rather similar to stable bimodules. The ideal completion also plays
a central r\^{o}le in Stone duality for distributive lattices \cite[\S
II]{johnstone_1982}.

\subsection*{Choice-free dualities}

It is well-known that many Stone-type dualities require the use of a choice-like
principle, e.g. the existence of prime ideals. Choice is sometimes only
necessary when connecting the localic viewpoint with the topological one; see
e.g. \cite[\S II.4]{johnstone_1982}.

Avoiding this use of choice has been a rather active area of research in recent
years, following the work of Bezhanishvili and Holliday
\cite{bezhanishvili_2020}. A choice-free duality for Heyting algebras, as well
as multiple references to recent literature, is given by Hartonas
\cite{hartonas_2024}.

\subsection*{Other related work}

Bezhanishvili et al. \cite{bezhanishvili_2024} present a positive modal logic.
Their semantics uses a meet-semilattice as a frame. Every formula is interpreted
as a filter over that, leading to the same falsity and disjunction clauses as
the ones I use here. However, the lack of joins and distributivity means that
they cannot handle implication. They also present some interesting links between
their logic and logics of \emph{independence} and \emph{team semantics}
\cite{yang_2016,kontinen_2016,yang_2017} to which the results of this paper
might be applicable.

De Groot and Pattison \cite{de_groot_2022} study the $(\land \times)$ fragment
of intuitionistic logic with a meet-preserving modality $\Box$. They give it a
semantics in semilattices, relating it to filters. Their semantics for $\Box$ is
based on relations which are extremely close to stable bimodules.

The coherent semantics appears rather close to the Kripke semantics of the
separating conjunction of the $\textsf{BI}$ logic of O'Hearn and Pym
\cite{ohearn_1999,pym_2002}. This is not surprising, as the Day convolution is
one of their main monoidal products. However, the fact that their Kripke
semantics can only be shown complete when falsity (interpreted as never true) is
excluded \cite[\S 4]{pym_2002} suggests that there might be interesting
connections with the results presented here.

Galal \cite{galal_2020} explores a categorification of the Scott-continuous
model of Linear Logic, which also consists of prime algebraic lattices (but with
weaker morphisms than the ones used here)
\cite{huth_1994,huth_2000,nygaard_2004,winskel_2004}. The key notion of
directed-completeness is replaced by sifted colimits. No connection to Kripke
semantics is made.

\begin{ack}

  I am grateful to Fabian Ruch for suggesting that product-preserving presheaves
  might be useful. Thanks are due to Daniel Gratzer for pointers to higher
  algebra; to Philip Saville and Pedro Azevedo de Amorim for the use of Day
  convolution in $\textsf{BI}$. I am also grateful to Younesse Kaddar, Marcelo
  Fiore, Ayberk Tosun, Jonathan Sterling, Sam Staton, and Alyssa Renata for many
  stimulating conversations and pointers to relevant literature. Finally, I am
  thankful to one of the anonymous reviewers for pointing out that the coherent
  semantics also work in a distributive semilattice, as well as multiple
  references to relevant literature on dualities.

\end{ack}

\bibliographystyle{./entics}
\bibliography{2dks}

\end{document}